\documentclass[12pt, onecolumn, draftclsnofoot,peerreview]{IEEEtran}

\usepackage{cite}
\usepackage{pslatex}
\usepackage{moreverb}
\usepackage{epsfig}
\usepackage{amsmath}
\usepackage{amssymb}
\usepackage{booktabs}
\usepackage{array}
\usepackage{multirow}
\usepackage{threeparttable}
\usepackage{url}
\usepackage{dsfont}
\usepackage{subfigure}
\usepackage{midfloat}

\begin{document}

\title{Diversity-Multiplexing  Tradeoff  in the Multiaccess Relay Channel with Finite Block Length}

\author{
    \IEEEauthorblockN{Chung-Pi Lee and Hsuan-Jung Su}\\
    \IEEEauthorblockA{Graduate Institute of Communication Engineering\\
    Department of Electrical Engineering\\
        National Taiwan University, Taipei, Taiwan\\
        Email: D96942016@ntu.edu.tw, hjsu@cc.ee.ntu.edu.tw}
}
\newtheorem{thm}{Theorem}
\maketitle

\begin{abstract}
The Dynamic Decode-and-Forward (DDF) protocol and the Hybrid DDF and
Amplified-and-Forward (HDAF) protocol for the multiple-access relay
channel (MARC) with quasi static fading are evaluated using the Zheng-Tse
diversity-multiplexing tradeoff (DMT). We assume that there are two users, one
half-duplex relay, and a common destination, each equipped with
single antenna. For the Rayleigh fading, the DDF protocol is well
known and has been analyzed in terms of the DMT with
\emph{infinite} block length. By carefully dealing with properties specific to \emph{finite} block length,   we characterize the
finite block length DMT which takes into account the fact that the
event of decoding error at the relay causes the degradation in error
performance when the block length is finite. Furthermore,
we consider the situation where the destination does
not
have a priori knowledge of  the relay decision time at
which the relay switches from listening to transmitting. By
introducing a decision rejection criterion
such that
the relay forwards message only when its decision is reliable, and
the generalized likelihood ratio test (GLRT) rule at the destination
that jointly decodes the relay decision time and the information
message, our analysis show that the optimal DMT is achievable as if
there is no decoding error at the relay and the relay decision time is
known at the destination. Therefore, infinite block length and
additional
overhead for communicating the decision time are
\emph{not} needed for the DDF to achieve  the optimal DMT. To
further improve the DMT, we propose the HDAF protocol which take
advantages of both the DDF and the Amplified-and-Forward protocols by
judiciously choosing which protocol to use. Our result shows that the
HDAF protocol outperforms the original DDF in the DMT perspective.
Finally, a variant of the HDAF protocol with lower implementation
complexity without sacrificing the DMT performance is devised.
\end{abstract}


\section{Introduction}
In recent years, cooperative communication has received significant
interest as a means of providing spatial diversity  when time,
frequency, antenna diversity are unavailable due to delay,
bandwidth and terminal size constraints, respectively. Cooperative
techniques  provide diversity by enabling users to utilize one
another's resources and has been extensively studied for the single
source from outage probability analysis or diversity-multiplexing
tradeoff (DMT) perspective \cite{cooperaive_outage},
\cite{cooperative_wirelss}, \cite{on_the_achievableDMT}.

Practical communication systems usually involve more than one user.
One of the most typical models is the multiple-access channel (MAC).
The capacity region of MAC is well known and the DMT is also
developed in \cite{DMT_MAC} for MAC. In \cite{on_the_achievableDMT},
\cite{dispacetime_protocol}, cooperative diversity was extended to
the multiple users cases. For most cooperative protocols,
substantial coordination among the users are required, which may be
impractical due to cost and complexity consideration. Alternatively,
we consider the multiple-access relay channel (MARC),i.e. the MAC
with a single shared relay and focus on the dynamic decode-and
forward (DDF) protocol \cite{on_the_achievableDMT}. For DDF
protocol, the relay does not decode until it is possible to
successfully decode source information message. The relay then
re-encodes the message and transmit it in the remaining coding
interval. In this model we concern, the users need not be aware of
the existence of the relay. All cost and complexity are placed in
the relay and destination. Such an architecture may be suitable for
infrastructure networks, where the relay and destination correspond
the station having more resource(i.e. base station). Moreover, sice
a single relay is shared by multiple users in the MARC,  the extra
cost of adding the relay per user may thus be more acceptable.
Finally, for further enhancement of DMT, we propose a hybrid
protocol which combines DDF and  multiple-access amplify-and-forward
(MAF) \cite{case_MAC}, \cite{DMT_forMARC} protocol to improve the
diversity gain at high multiplexing gain region.

\subsection{Related Research}
The MARC was first introduced in \cite{MARCfirst}. In MARC, the
relay helps multiple sources simultaneously to reach a common
destination. Information-theoretic treatment of the MARC has focused
on two aspects, namely, the capacity region and the DMT (the outage
behavior of slow fading channel in the high signal-to-noise (SNR)
regime \cite{diva}). The achievable rate for the MARC has been
proposed in \cite{Capcity_MARC}, \cite{cooperaive_capacity},
\cite{multiple_mul}. However, the capacity region of general MARC
remains unknown. The DMT for the half-duplex MARC with single
antenna nodes is studied in \cite{case_MAC}, \cite{DMT_forMARC},
\cite{optimality_ARQ}, \cite{cooperative_wirelss}. In
\cite{case_MAC}, \cite{DMT_forMARC}, it is shown that MAF protocol
is DMT optimal for high multiplexing gains; however, this protocol
remains to be suboptimal for low multiplexing gains compared with
DDF strategy \cite{optimality_ARQ}. In addition, the extra overhead
to communicate the channel realizations of source-relay links to
destination is required for MAF. Another relaying strategy for the
MARC is compress-forward (CF) \cite{cooperative_wirelss}. In CF, the
relay exploits Wyner-Ziv coding to compress its received signal and
forward it to destination. It has been shown in
\cite{cooperative_wirelss} that CF also acheves the optimal DMT for
high multiplexing gains but suffers from diversity loss for low
multiplexing gains,  moreover,  CF costs much larger complexity.

The present paper focuses  on the DDF protocol for the half-duplex,
single relay, single antenna case due to its nice balance between
complexity and good performance at low multiplexing gain. Moreover,
we propose a hybrid DDF-MAF (HDAF) protocol to enhance  the  DMT by
improving the poor performance of DDF at high multiplexing gain.

\subsection{Summary of Results}
Previous work, \cite{optimality_ARQ}, \cite{on_the_achievableDMT}
assume an infinite block length such that there is no decoding error
at the relay  and the number of relay decision moments is also
assumed infinite.  Inspired by \cite{coding_and_decoding_DDF} for
the single user, single relay case, we analyze explicitly the
achievable DMT with finite block length and finite relay decision
moments for the MARC. Note  the former is a special case of  the
MARC when there is only one user. Moreover, different from the proof
of DMT achievability in \cite{coding_and_decoding_DDF}, we do not
separately average the probability of event at relay and destination
over the random codebook since they are not fully independent, which
will be made more clear in Section \ref{sec.avg}. As
\cite{coding_and_decoding_DDF}, two issues are discussed in our
model: 1) the effect of decoding errors at the relay. 2) the fact
that the relay decision time is not generally known priori at
designation. In order to tackle 1) a decision rejection criterion at
the relay, such that the relay triggers transmission only when its
decoding is reliable. For 2), the destination  jointly decodes the
relay decision time and the information messages based on the
generalized likelihood ratio test (GLRT) rule.  Our results show
that in order to achieve the DMT,   additional protocol overhead
informing the destination about the relay decision time is not
necessary and the loss of DMT due to decoding error at the relay can
be avoided. Finally, we propose HDAF protocol which takes advantage
of both DDF and MAF protocol. Our analysis shows that with the
\emph{finite} block length $MT$, HDAF protocol outperforms the
original DDF protocol, especially at high multiplexing gain region
and low multiplexing region  when $M$ is small. HDAF also has better
DMT than MAF protocol  when $M$ is moderate large. In addition,
without causing any loss in DMT perspective, a variant of HDAF with
lower complexity by allowing the relay's transmission only after
$M/2$ relay decision moments \cite{cooperaive_latti_half} is
devised. Notice that the analysis with finite $T$ will be much
complicated since we have to deal with events which depends not only
on channel realization(i.e outage)  but also on codebooks when
applying random coding techniques.

In Section \ref{sec.model}, we introduce the system model and review
relevant previous results. In Section \ref{sec.mainresult},  a
characterization of the DMT of the DDF protocol for MARC with finite
block length is  presented. In Section \ref{sec.hybridHDAF}, we
devise HDAF and its variant and characterize the DMT with finite
block length. Section \ref{sec.conclusion} gives the conclusion.

%

\section{Notations and Definitions} \label{sec.model}

\subsection{System Model}
\begin{figure}[t] 
\centering
\includegraphics[scale=1, width=0.6\textwidth]{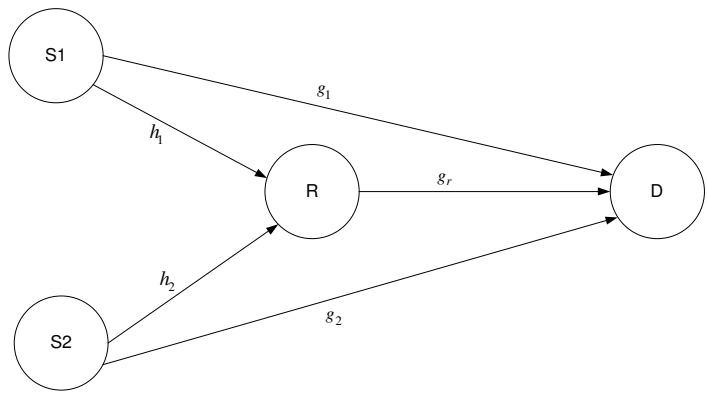}
\caption{MARC model} 
\label{fig.1}
\end{figure}

 We Consider the two-user(S1 and S2) MARC model, a relay node (R)
is assigned to assist the two multiple access users, (see
Fig.\ref{fig.1}). The users are not allowed to help each other (due
to practical limitations, for example). The relay node is
constrained by the half-duplex assumption, i.e. the relay can not
transmit and receive simultaneously. Each node equips with single
antenna.

All wireless links are assumed to be frequency nonselective and
block fading, where the channel coefficients are random but remain
constant over the whole duration of a codeword and the channel
coherence time is much larger than the allowed decoding delay. Let
$h_i$, $g_i$, $g_r$ denote the fading coefficients between the user
${i}$ to relay, user ${i}$ to destination, and relay to destination,
respectively. The channel fading coefficients are i.i.d.
$\mathcal{CN}$(0,1) variables, corresponding to i.i.d. Rayleigh
fading. Here we assume the perfect channel state information at
receiver (CSIR). We adopt the slotted transmission where a codeword
spans $M$ slots and  each slot consists of length $T$ symbols. There
are thus a total block length of $MT$.
 In decode-and-forward protocols, the block of length $MT$ symbols
is split into two phases. In the first phase, the relay receives the
signal from the source until the end of a certain slot, referred as
the decision time. Then  the relay tries to decode the source
message. In the second phase, based on the decoded message, the
relay sends its codeword to the destination in the remaining block .
For DDF protocol, the decision time depends  on the channel
coefficient and the received signal. For the first phase, the signal
received by the relay is
\begin{equation}
y_{r,k}=\sum^{2}_{i}h_{i}x_{i,k}+n_{k},  \;\;\;\;\;
k=1,2,...,\textbf{m}T \label{eq.2}
\end{equation}
Note  the decision time at the end of slot \textbf{m} is a random
variable.  The signal received by the destination is
\begin{equation}
y_{d,k}=\sum^{2}_{i}g_{i}x_{i,k}+v_{k},  \;\;\;\;\;
k=1,2,...,\textbf{m}T \label{eq.3}
\end{equation}
During the second phase, the signal received by the destination is
\begin{equation}
y_{d,k}=\sum^{2}_{i}g_{i}x_{i,k}+g_{r}x_{r,k}+w_{k},  \;\;\;\;\;
k=\textbf{m}T+1,\textbf{m}T+2,...,MT \label{eq.4}
\end{equation}
We let $\textbf{x}_i=[x_{i,1},...x_{i,MT}]^T$ denote the user $i$
codeword with rate $R_i$ bits per symbol. $\textbf{x}_r$ is
similarly defined but only the last $(M-\mathbf{m})T$ symbols are
transmitted. The same average power constraint $P$ per symbol are
imposed on each user and relay,
 \[
 E[|x_{i,k}|^{2}]\leq P, \;\; E[|x_{r,k}|^{2}]\leq P
  \]
where $E[\;]$ denotes expectation over all codewords. The noise at
the relay and destination  are independent Gaussian noise with
variance of $\sigma^{2}_{n}$, $\sigma^{2}_{v}$, denoted as
$n_{k}\sim \mathcal{CN}(0,\sigma^{2}_{n})$ and $w_{k}\sim
\mathcal{CN}(0,\sigma^{2}_{v})$ respectively.
$\rho=\frac{P}{\sigma^{2}_{v}}$ and $\rho'=\frac{P}{\sigma^{2}_{n}}$
define the SNRs of the source-destination and the source-relay links
,respectively. In our model, we consider the symmetric case where
$R_{1}=R_{2}=R/2$, the sum rate is $R$ bits per symbol. As
\cite{optimality_ARQ}, the relay decision time is chosen to be the
instant that the transmitted rate is within the achievable rate
region at the end of $\mathbf{m}$ slot, which satisfies
\begin{align}
&MTR_{1} < m\log(1+|h_{1}|^{2}\rho')T     \label{eq.mac1} \\
&MTR_{2} <m\log(1+|h_{2}|^{2}\rho')T        \label{eq.mac2}  \\
&MTR <m\log\left(1+(|h_{1}|^{2}+|h_{2}|^{2})\rho'\right)T
\label{eq.mac12}
\end{align}
$\mathbf{m}$ is set to the minimum $m=1,2,...,M-1$ such that
(\ref{eq.mac1}), (\ref{eq.mac2}), (\ref{eq.mac12}) hold, otherwise,
$\mathbf{m}=M$  and the relay remains silent.

For later use,  we denote the complement of an event $A$ by
$\overline{A}$, the transpose and  Hermitian transpose  of
$\mathbf{z}$ by $\mathbf{z}^{T}$, $\mathbf{z}^{H}$, respectively.
$(x)^{+}= x$ if $x>0$, otherwise equal to zero. Let
$\textbf{x}^{k}_{i,n},$ and $\textbf{x}^{k}_{r,n},$ denote the
source transmit signal from time $nT+1$ to $kT$ and the relay
transmit signal from time $nT+1$ to $kT$ respectively.
$\textbf{y}^{k}_{r,n},$, $\textbf{x}^{k}_{d,n},$ are similarly
defined.

\subsection{Diversity-Multiplexing Tradeoff (DMT)}
Our work use a lot the notion of DMT posed in \cite{diva}. We only
provide definitions here. Consider a family of codes, such that the
code has a rate of $R(\rho)$, corresponding to SNR $\rho$ bits per
channel use(BPCU) and error probability $P_{E}(\rho)$ The
multiplexing gain $r$ and the diversity gain  defined as
\begin{equation}
r \triangleq \lim_{\rho \rightarrow \infty} \frac{R(\rho)}{\log
\rho} , \;\; d \triangleq -\lim_{\rho \rightarrow \infty} \frac{\log
P_{E}(\rho)}{\log \rho}  \label{eq.dmt}
\end{equation}
and we can write as $P_{E}(\rho)\doteq \rho^{-d}$, where $\doteq$
denotes the exponential equality . $\dot{\leq}$ and $\dot{\geq}$ are
similarly defined. For the point to point multiple-input
multiple-output(MIMO) channel with $m$ transmit and $n$ receive
antennas, $r \leq \min(m,n)$, the optimal $d^{*}(r)$ is referred to
as DMT.

\section{DMT of the DDF protocol for MARC with finite block length}  \label{sec.mainresult}

\subsection{Effect of Finite Block Length}
The DMT of DDF protocol for MARC with symmetric rate
$R_{1}=R_{2}=\frac{r}{2}\log \rho$ has been shown in
\cite{optimality_ARQ},
\begin{equation}
d^{*}_{DDF-MARC}(r)=\begin{cases}
2-r   \;\;\;\;\;\;\;\;  \frac{1}{2}>r \geq 0 \\
3(1-r) \;\; \frac{2}{3}>r \geq \frac{1}{2} \\
2\frac{1-r}{r}   \;\;\;\;\;\;\;\; 1 \geq r \geq \frac{2}{3}
\end{cases} \label{eq.ddf}
\end{equation}
However, to achieve the DMT in (\ref{eq.ddf}), a scheme with
$M\rightarrow\infty$ possible decision times is
necessary\cite{optimality_ARQ}, furthermore, an infinite block
length $T \rightarrow \infty $ is assumed such that there is no
decoding error at the relay. For the practical  code design, the
code length is finite and the error event of decoding at the relay
occurs even though the transmitted rate falls in the achievable rate
region, i.e. (\ref{eq.mac1}),  (\ref{eq.mac2}),  (\ref{eq.mac12})
are satisfied. Forwarding the wrong source information message would
significantly degrade the error performance at the destination. Thus
the probability of relay decoding error dominates the error
probability at the destination. The DMT analysis with finite $M$ and
$T$ has been treated in \cite{coding_and_decoding_DDF} for the relay
case, a special case of  MARC model  if there is only one user. In
our MARC model, there are two users interfering with each other,
thus the derivation of DMT is more involved, moreover, we take an
approach different from \cite{coding_and_decoding_DDF} to averaging
error probability over all the random codebook.

\subsection{Characterization of DMT with Finite Block Length}
In this section, the proof uses the machinery
of \cite{coding_and_decoding_DDF}, \cite{optimality_ARQ},
\cite{DMT_MAC},\cite{on_the_achievableDMT}. Therefore, we only
provide a sketch of the previous results involved and focus on the
novel parts. First, we find an upper bound on the DMT by letting
$T\rightarrow \infty $, and assume the destination has the knowledge
of $\mathbf{m}$. By relaxing the constraint that $T$ is finite, we
characterize the DMT using outage probability analysis, hence make
it as an upper bound for the finite block length. This is
established by the following theorem.
\begin{thm}\label{theorem upperbound}
The DMT of the two users, single-relay, single destination DDF
scheme with decision times $m=1,2,...,M$ and finite slot length $T$
is upper-bounded by
\begin{equation}
d_{out}(r)=\min_{1\leq m\leq
M}\left\{{d}_{m,R}(r)+d_{m,D}(r)\right\} \label{eq.44}
\end{equation}
 \end{thm}
where $d_{m,R}(r)$ is defined in (\ref{eq.30}), and $d_{m,D}(r)$ is
defined in (\ref{eq.40})-(\ref{eq.43}).
\begin{proof}
We  use similar techniques developed in \cite{DMT_MAC},
\cite{coding_and_decoding_DDF}. Let $P_{out}(r)$ denote the outage
probability at the destination.  $P_{out}^{m}(r)$ denote the outage
probability at the destination for a given $\mathbf{m}=m$. The
outage event $\mathcal{O}^{m}_{D}$ for a given $\mathbf{m}=m$ is
defined as
\begin{equation}
\mathcal{O}^{m}_{D}=\left\{(g_{1},g_{2},g_{r}):
 \left(\mathcal{O}^{m}_{1,D} \bigcup
\mathcal{O}^{m}_{2,D}\bigcup
\mathcal{O}^{m}_{(1,2),D}\right)\right\}
\end{equation}
where $\mathcal{O}^{m}_{1,D}$, $\mathcal{O}^{m}_{2,D}$,
$\mathcal{O}^{m}_{(1,2),D}$ are defined as
\begin{equation}
\begin{split}
\mathcal{O}^{m}_{1,D}=& \left\{(g_{1},g_{r}):
 mT\log\left(1+|g_{1}|^{2}\rho\right) \right. \\
&+\left.(M-m)T\log\left(1+\left(|g_{1}|^{2}+|g_{r}|^{2}\right)\rho\right)\leq
MTR_{1}\right\}
\end{split} \label{eq.31}
\end{equation}
$\mathcal{O}^{m}_{2,D}$ is similarly defined by replacing $1$ with
$2$.
\begin{equation}
\begin{split}
\mathcal{O}^{m}_{(1,2),D}=&\left\{(g_{1},g_{2},g_{r}):
mT\log\left(1+\left(|g_{1}|^{2}+|g_{2}|^{2}\right)\rho\right) \right.\\
&
+\left.(M-m)T\log\left(1+\left(|g_{1}|^{2}+|g_{2}|^{2}+|g_{r}|^{2}\right)\rho\right)\leq
MTR\right\}   \label{eq.32}
\end{split}
\end{equation}

Then write
\begin{equation}
P_{out}(r)=\sum_{m=1}^{M}P(\mathbf{m}=m)P^{m}_{out}(r)
\label{eq.upper}
\end{equation}

Since scaling SNR by a constant $(\rho'/\rho)$ does not change the
DMT. Define
\begin{equation}
\begin{split}
P_{out}(r) & \doteq \rho^{-d_{out}(r)}\\
P_{out}^{m}(r) & \doteq \rho^{-d_{m,D}(r)},\;\;\; 1\leq m \leq M \\
P(\mathbf{m}=m) & \doteq \rho^{-d_{m,R}(r)},\;\;\; 1\leq m \leq M
\end{split} \label{eq.appendixupper}
\end{equation}
Then (\ref{eq.44}) clearly follows from (\ref{eq.upper}) The proof
of (\ref{eq.appendixupper}) is provided in Appendix
\ref{appendixupper}. It remains to show that it is an upper bound of
finite block length case. From standard arguments based on Fano's
inequality \cite{DMT_MAC},  it can be seen that $P_{out}^{m}(r)$ is
indeed the best we can get, thus completes the proof.
\end{proof}

Note in the proof above, the relay decision time $\mathbf{m}$ is
assumed priori known  at the destination. The following theorem
shows that this assumption is not necessary and the upper bound  is
achievable with finite $T$.

\begin{thm}\label{theroem achivable}
The  upper bound of Theorem \ref{theorem upperbound} is achievable
for finite-length $T$ and no priori knowledge of decision time at
the destination decoder.

\end{thm}
\begin{proof}
In order to prove the achievability, we use standard random coding
argument with bounded distance decoder \cite{optimality_ARQ} at the
relay to overcome the effect of relay decoding error
\cite{coding_and_decoding_DDF}. Owing to the introduction of the
bounded distance decoder at the relay, the probability of relay's
decision time at $m$ slot and  the probability of decoding error at
destination  given relay decision time $m$ are not independent(the
relay's decision time not only depends on the source-relay links but
also on the codewords). Therefore,  different from
\cite{coding_and_decoding_DDF} where these two terms are
\emph{separately } averaged over the random ensemble, we take an
approach to directly averaging the resulting error probability at
the destination over the random ensemble.

\emph{Codebook Generation}: For given $M$, $T$,
$R_{1}=R_{2}=\frac{r}{2}\log \rho$, according to i.i.d components
$\mathcal{CN}(0,P)$, independently generate $\rho^{\frac{rMT}{2}}$
codewords, $\mathcal{X}_{i}\subset\mathbb{C}^{MT}$, for each $i=1,2$
  and $\mathcal{X}_{r}\subset\mathbb{C}^{MT}$ of cardinality
 $\rho^{rMT}$. We let $\bold{x}_{i}(w_{i})$, $\bold{x}_{r}(w)$
denote the codewords in $\mathcal{X}_{i}$ and in $\mathcal{X}_{r}$
respectively, corresponding to the information message $w_{i} \in
\{1,...,\rho^{\frac{rMT}{2}} \} $, $w \in \{1,...,\rho^{{rMT}} \} $.

\emph{Relay Decoding:}  From (\ref{eq.mac1})-(\ref{eq.mac12}), we
define the relay outage event at slot $m$ as
\begin{equation} \label{eq.45}
\mathcal{O}_{R}^{m}=\left \{ (h_{1},h_{2}):
\left(\mathcal{O}_{1,R}^{m}\bigcup \mathcal{O}_{2,R}^{m}\bigcup
\mathcal{O}_{(1,2),R}^{m}\right)\right\}
\end{equation}
and
\begin{equation}
\mathcal{O}_{1,R}^{m}\triangleq\left\{h_{1} : |h_{1}|^{2} \leq
\frac{\rho^{\frac{\frac{r}{2}M}{m}}-1}{\rho'}\right\}
\label{eq.relaymac1}
\end{equation}
$\mathcal{O}_{1,R}^{m}$ is similarly defined.
\begin{equation}
\mathcal{O}_{(1,2),R}^{m} \triangleq
\left\{(h_{1},h_{2}):|h_{1}|^{2}+|h_{2}|^{2} \leq
\frac{\rho^{\frac{rM}{m}}-1}{\rho'}\right\} \label{eq.relaymac12}
\end{equation}
Due to the finite block length $MT$, the relay may decode in error
even $(h_{1}, h_{2}) \notin \mathcal{O}_{R}^{m}$. Then an incorrect
codeword is sent to the destination, causing significant
interference. Thereofre, similar to \cite{mimoarq},
\cite{coding_and_decoding_DDF} for the relay case, we introduce a
bounded distance relay decoding decision function $\psi_{\delta}$
defined as follows: for $m=1,...,M-1$, define the regions
$\mathcal{S}_{m}(w_{1},w_{2})$ consisting of all points $\bold{y}\in
\mathbb{C}^{mT}$ for which $(w_{1},w_{2})$ is the unique message
enclosed in a sphere of squared radius $mT(1+\delta)\sigma_{n}^{2}$
centered at $\bold{y}$, i.e.,
$|\bold{y}-h_{1}\bold{x}_{1,0}^{m}(w_{1})-h_{2}\bold{x}_{2,0}^{m}(w_{2})|^{2}\leq
mT(1+\delta)\sigma_{n}^{2}$. If \\
1) $(h_{1},h_{2})\notin \mathcal{O}^{m}_{R}$, \\
2) $\bold{y}_{r,0}^{m} \in
\mathcal{S}_{m}(\hat{w}_{1},\hat{w}_{2})$,
\\
Then, $\psi_{\delta}(\bold{y}_{r,0}^{m},h_{1},h_{2})$
 outputs the decoded message $\hat{w}$, where
$\hat{w}=(\hat{w}_{1},\hat{w}_{2})$  and the relay starts to
transmit the signal $\mathbf{x}_{r,m}^{M}(\hat{w})$  for the
remaining part of the block. Otherwise, it waits for the next slot.

\emph{Destination Decoding : } The destination is not aware of the
relay decision time $\bold{m}$, hence it simultaneously detects the
decision time and the information message according to the GLRT
rule:\\
\begin{equation} \label{eq.46}
\{\hat{w},\hat{m}\}=\arg \max_{w,m}p(\mathbf{y}^{M}_{d,0}|w,m)
\end{equation}
where $p(\mathbf{y}^{M}_{d,0}|w,m)$ is the decoder likelihood
function.

\emph{Error Probability Analysis : } Let $E$ denote the decoding
error event at the destination and $E_{r}$ denote the decoding error
event at the relay. Follow the steps in
\cite{coding_and_decoding_DDF}, we have the following results,
\begin{align}
P_{E}=&\sum_{m=1}^{M} P(\mathbf{m}=m)P(E|\mathbf{m}=m) \label{eq.47} \\
\leq & \sum_{m=1}^{M} P(\mathbf{m}=m)
\left(P(E_{r}|\mathbf{m}=m)+P(E|\overline{E}_{r},\mathbf{m}=m)
\right) \label{eq.48}
\end{align}

For $m=1,...,M-1$,
let $\delta=\mu \log{ \rho}$, we have
\begin{equation}
P(E_{r}|\bold{m}=m)\dot{\leq}\rho^{-mT\mu}
\end{equation}
Since
$P(E|\overline{E}_{r},\mathbf{m}=m)\dot{\geq}\rho^{-d_{m,D}(r)}$, we
can choose a sufficiently  large finite $\mu$ to make the terms
$P(E_{r}|\bold{m}=m) $ exponentially irrelevant in (\ref{eq.48}),
i.e. since $d_{m,D}(r)<3$, we can choose $\mu T >3$. As for the
other terms in (\ref{eq.48}), as mentioned earlier, they are not
independent event since the relay's decision time not only depends on
the source-relay links but also on the codewords. In Appendix
\ref{appenddixlower}, unlike \cite{coding_and_decoding_DDF}, by
averaging these two terms together, we show the average probability
of the term $P(\mathbf{m}=m)P(E_{r}|\mathbf{m}=m)$ using Gaussian
random ensemble is exponentially upper bounded by
 $d_{out}(r)$, thus completes the proof.

\end{proof}
It appears intractable to obtain a closed form  of (\ref{eq.44}). In
Fig.\ref{fig.2}, $d_{out}(r)$ are plotted for $M=2,5,,10,20$.
\begin{figure}[t] 
\centering
\includegraphics[scale=1, width=0.6\textwidth]{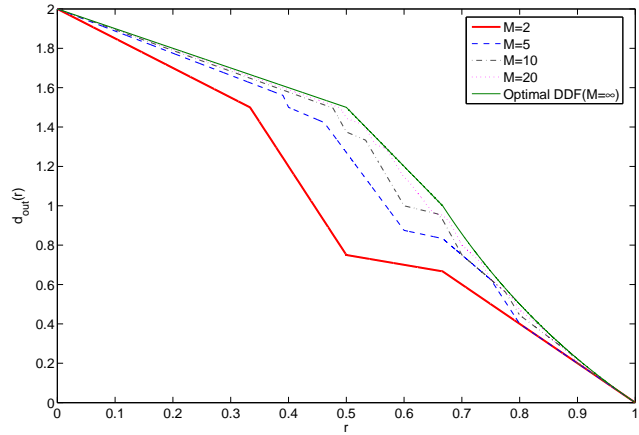}
\caption{The DMT of the DDF protocol for MARC with finite $M$} 
\label{fig.2}
\end{figure}
As $M$ grows,  $d_{out}(r)$  is seen to approach  the optimal DMT
in\cite{optimality_ARQ} where $M\rightarrow\infty$. Fig.\ref{fig.2}
also shows that even for a moderate value of $M$, $d_{out}(r)$
achieves the DMT close to  optimal one.
\section{Hybrid Amplified-forward and decode-forward protocol}
\label{sec.hybridHDAF}
\subsection{The DMT of Hybrid Amplified-Forward and Decode-Forward protocol }
It has been reported in \cite{case_MAC}
\cite{DMT_forMARC}, the DMT of the MAF protocol is given by
\begin{equation}
d_{MAF}(r)=
\begin{cases}
2-\frac{3r}{2}, \;\;\;\; \;\;\; 0\leq r \leq \frac{2}{3} \\
3(1-r),  \;\;\; \frac{2}{3}\leq r \leq 1
\end{cases}
\label{eq.mafdmt}
\end{equation}
The optimal diversity gain for high multiplexing gain
($\frac{2}{3}\leq r \leq 1$) is achieved by MAF protocol. On the
contrast, for the DDF strategy, the relay will only be able to
during a small fraction of time slots, hence suffers form a loss
compared to MAF. Combined with MAF protocol, the DMT may be further
improved.  This motivates us to  propose a new hybrid
strategy as follows: \\
\emph{(1)} For $ \frac{2}{3}< r \leq 1$, the relay
simply uses the MAF protocol. \\
\emph{(2)} For $ \frac{1}{2}< r \leq \frac{2}{3}$, the relay
simply uses the DDF protocol. \\
\emph{(3)} For $ 0 \leq r \leq \frac{1}{2}$,  the relay dynamically
decode the source messages before $M/2$ (including $M/2$) time slots
,where we assume M is even. If it can not successfully decode , then
the MAF protocol is instead used after the time slot $M/2$.

The goal of the hybrid strategy aims to take  advantages of both MAF
and DDF protocol, we refer it as HDAF protocol. Note HDAF has the
same setting as Section \ref{sec.mainresult} (i.e. finite block
length) when DDF is chosen to be used.  We have the following
theorem and Fig.\ref{fig.2} showing that HDAF outperforms DDF
especially at $\frac{2}{3}< r \leq 1$ and $ 0 \leq r \leq
\frac{1}{2}$ when the number of time slots, namely $M$ is small
since HDAF can achieve the optimal DMT of DDF protocol with infinite
$M$ for $0\leq r \leq \frac{1}{2}$ even a finite $M$ time slots
being used. Furthermore, HDAF  has better diversity gain than MAF
for $ \frac{1}{2}< r \leq \frac{2}{3}$ when $M$ is large enough
since it will approach close to the DMT of DDF with infinite $M$.
\begin{thm} \label{theoremhybrid}
The DMT of HDAF protocol is
\begin{equation}
d^{HDAF}(r)=
\begin{cases}
2-r \;\;\;\;\;\;\;\;\;\, 0\leq r \leq \frac{1}{2}\\
d_{out}(r)\;\;\;\;\;\; \frac{1}{2}< r \leq \frac{2}{3}\\
3(1-r) \;\;\;\;\frac{2}{3}< r \leq 1
\end{cases}
 \label{eq.hybdmt}
\end{equation}
\end{thm}
\begin{proof}
To characterize the DMT of HDAF protocol $d_{HDAF}(r)$, for $
\frac{2}{3}< r \leq 1$, the MAF protocol will be used and
$d_{HDAF}(r)=d_{MAF}(r)$. The advantages of MAF for high
multiplexing gain is preserved.
For $ \frac{1}{2}< r \leq \frac{2}{3}$, the DDF protocol will be
used and $d_{HDAF}(r)=d_{out}(r)$ in (\ref{eq.44}), hence this
region entails no loss in terms of DMT compared to DDF protocol. For
$ 0 \leq r \leq \frac{1}{2}$, from  \cite{optimality_ARQ}, we know
$2-r$ is indeed the upper bound for MARC, we will show that
$d_{HDAF}(r)\geq 2-r$, thus establish $d_{HDAF}(r)=2-r$. First
consider the case where source message is decoded at the relay
before $M/2$ time slots. In this case,  since the upper bound of
(\ref{eq.44}) is achievable by finite code length, we can similarly
write $d_{HDAF}(r)$ as
\begin{equation}
d_{HDAF}(r)=\min_{1\leq m\leq
M/2}\left\{{d}_{m,R}^{HDAF}(r)+d_{m,D}^{HDAF}(r)\right\}
\label{eq.hyb1}
\end{equation}  It clear that
$d_{m,R}^{HDAF}(r)=d_{m,R}(r)$, $d_{m,D}^{HDAF}(r)=d_{m,D}(r)$ for
$1\leq m \leq M/2$.  From (\ref{eq.40}),  the diversity gain at
least $2-r$ is obtained since $d_{m,D}^{HDAF}(r)=2-r$. However, this
is not necessarily the case since the source message may  be decoded
only after $M/2$ time slots and the MAF protocol will be used
instead. It remains to derive the DMT of  HDAF protocol for the case
where decoding moment  $\mathbf{m}>M/2$. In \cite{case_MAC}, the
error provability $P_{e}$ using MAF protocol is shown to be
upper-bounded by
\begin{equation}
P_{e} \leq P_{e_{1}}+P_{e_{2}}+P_{e_{(1,2)}} \label{eq.hyb9}
\end{equation} where
$P_{e_{I}}$, $I=1,2,(1,2)$ represent the probability of the error
event that user(s) $I$  are detected in error. $P_{e_{I}}$'s  are
averaged over the ensemble of channel realizations, and thus leads
to lower diversity gain $2-\frac{3r}{2}$  for $0\leq r\ \leq
\frac{2}{3}$. However, for our case, the MAF protocol is used only
when the relay can not decode the source message before $M/2$ time
slots and the resulting diversity gain is expected to be larger. To
deal with the HDAF case, from (\ref{eq.51}), the probability that
the source message can not be decoded before $M/2$ time slots can be
split into two events
\[
 \left\{(h_{1},h_{2})\in
\mathcal{O}_{R}^{\frac{M}{2}}\right\}\bigcup
\left\{\{(h_{1},h_{2})\notin
\mathcal{O}_{R}^{\frac{M}{2}}\},\{\bold{y}_{r,0}^{\frac{M}{2}}\notin
\mathcal{U}_{\frac{M}{2}}\}\right\}
\]
Then we need to upper bound the  probability that MAF is used and
decoding error occurs at the destination,   defined as
\begin{equation}
\begin{split}
P_{E,MAF}&=P\left(E, \left\{(h_{1},h_{2})\in
\mathcal{O}_{R}^{\frac{M}{2}}\right\}\right)\\
&+P\left(E,\left\{\{(h_{1},h_{2})\notin
\mathcal{O}_{R}^{\frac{M}{2}}\},\{\bold{y}_{r,0}^{\frac{M}{2}}\notin
\mathcal{U}_{\frac{M}{2}}\}\right\}\right)
 \end{split}  \label{eq.hyb3}
\end{equation}
 For the infinite block
code length case, only the first term on the right hand side
(\ref{eq.hyb3}) need to be considered. Nevertheless, in our finite
block code length case , the second term need to be taken into
account which much complicates the analysis. We upper bound the
first term here and leave the analysis of the second term in
Appendix \ref{appendixhybrib}.

Similar to (\ref{eq.hyb9}), partition the first term on the right
hand side of (\ref{eq.hyb3}) into three mutually exclusive error
events. To ease the notation, we denote them as $P_{e_{I}}$ in this
Section. Define $|h_{i}|^{2}=\rho^{-\alpha_{i}}$,
$|g_{i}|^{2}=\rho^{-\beta_{i}}$, where $i=1,2,r$. From
\cite{on_the_achievableDMT}, the $P_{e_{1}}$ conditioned on $h_{i}$,
$g_{i}$ is
\begin{equation}
P_{e_{1}|\alpha_{1},\beta_{1},\beta_{r}}\dot{\leq}
\rho^{-\frac{MT}{2}\left[\left(\max\left\{2(1-\beta_{1}),
1-(\beta_{r}+\alpha_{1})\right\}\right)^{+}-r\right]}
 \label{eq.hyb2}
\end{equation}
and define the outage event as
\begin{equation}
\begin{split}
O_{1,HDAF}^{+} = \left\{(\alpha_{i},\beta_{i}) \in
\left\{\mathbb{R}^{5+}\bigcap \left(O_{1,R}^{\frac{M}{2}}\bigcup
O_{2,R}^{\frac{M}{2}}\bigcup
O_{(1,2),R}^{\frac{M}{2}}\right)\right\}: \right. \\
\left. \left(\max\left\{2(1-\beta_{1}),
1-(\beta_{r}+\alpha_{1})\right\}\right)^{+}\leq r\right\}
\end{split}
\label{eq.hyb13}
\end{equation}
From (\ref{eq.relaymac1})-(\ref{eq.relaymac12}),
\[
O_{i,R}^{\frac{M}{2}}=\left\{\alpha_{i}: (1-\alpha_{i})^{+}\leq
\frac{r}{2}\right\}  \;\;\;\;i=1,2
\]
\[
O_{(1,2),R}^{\frac{M}{2}}=\left\{(\alpha_{1},\alpha_{2}):
(1-\min(\alpha_{1}, \alpha_{2}))^{+}\leq r \right\}
\]
Let $P_{e_{i}}\doteq\rho^{d_{P_{e_{i}}}(r)}$, by discussing cases of
(\ref{eq.hyb13}) for different events of  $O_{I,R}^{\frac{M}{2}}$ ,
it can be shown a diversity gain at least $2-r$ is obtained.
\begin{equation}
d_{P_{e_{1}}}(r)=\inf_{O_{1,HDAF}^{+}}\left\{\alpha_{1}+\alpha_{2}+\beta_{1}+\beta_{2}+\beta_{r}\right\}\geq
2-r \label{eq.hyb4}
\end{equation}
$d_{P_{e_{2}}}(r)$ can be derived  in identical manner. For
$d_{P_{e_{(1,2)}}}(r)$,  the DMT exponent conditioned on
$h_{1},h_{2}$  \cite{case_MAC}, is
\begin{equation}
\begin{split}
&d_{P_{e_{(1,2)}}|(h_{1},h_{2})}(r)\\=&
\begin{cases}
2(1-r)^{+}  \;\;\;\;\;\;\;\;\;\;\;\;\;\;\;\;\;\;\;\;\;\;\;\;\;\;\;\;\;\;\; \min\{\alpha_{1}, \alpha_{2}\}>(1-r)^{+}\\
\left[3(1-r)-\min\{\alpha_{1}, \alpha_{2}\}\right]^{+}\;0\leq
\min\{\alpha_{1}, \alpha_{2}\}\leq(1-r)^{+}
\end{cases}
\end{split}\label{eq.hyb5}
\end{equation}
Then
\begin{equation}
d_{P_{e_{(1,2)}}}=\inf_{(\alpha_{1},\alpha_{2}) \in
O_{R}^{\frac{M}{2}}}
\left\{\alpha_{1}+\alpha_{2}+d_{P_{e_{(1,2)}}|(h_{1},h_{2})}(r)\right\}>(2-r)
\label{eq.hyb8}
\end{equation}
\end{proof}
Theorem \ref{theoremhybrid} shows that even for small $M$, HDAF can
still achieve the optimal DMT of DDF protocol with infinite $M$ when
$0\leq r \leq \frac{1}{2}$. The loss in DMT of finite $M$ compared
to infinite $M$ case comes from  the fact that the event of
erroneous decoding at the destination when the relay starts
transmission after $M/2$ time slots becomes the dominant error
event, which will not occur in the HDAF strategy since MAF is used
instead in that case.
\subsection{Variant of HDAF protocol}
Indicated in \cite{cooperaive_latti_half}, allowing the relay node
to start transmission at any time slots may result in higher
complexity since this requires a very high-dimensional constellation
to ensure the possibility of source message being uniquely decodable
within a small code length. We can also prove  the modified HDAF
protocol that allowing the relay to transmit \textsl{only} after the
$M/2$ time slots do not entail any loss in the DMT perspective.
\begin{thm}
The modified HDAF protocol still achieves the DMT of Theorem
\ref{theoremhybrid}.
\end{thm}
\begin{proof}
Note DDF protocol is  used only for $0\leq r\leq \frac{2}{3}$.
Replace (\ref{eq.51}) with $1$ to   upper bound $P(\mathbf{m}=m)$
for all $1\leq m \leq M/2$, i.e. relay decodes before $M/2$ time
slots. Following the remaining steps thereafter, we can obtain an
upper bound $d_{m,D}(r)$. Observe (\ref{eq.40})-(\ref{eq.42}),
$d_{m,D}(r)$ is already equal to optimal value $2-r$, $3(1-r)$ for
$0\leq f \leq \frac{1}{2}$, $0\leq f \leq \frac{5r-2}{2(3r-1)}$,
respectively. Since $\frac{5r-2}{2(3r-1)}\geq \frac{1}{2}$ for
$\frac{2}{3} \geq r \geq \frac{1}{2}$, decoding at $f\geq
\frac{1}{2}$ is good enough, we conclude that the modification does
not affect the DMT achieved by the  HDAF protocol.
\end{proof}


\section{Conclusion}  \label{sec.conclusion}
In this paper, we consider the design of cooperative protocol for
MARC consisting half-duplex nodes. We analyze the DMT of DDF and
proposed HDAD for MARC model with finite block length. The analysis
captures the practical issue that the relay may decode erroneously
even when the channel is not outage and the destination  have no
priori knowledge of the relay's decision moment.  Our results show
that  additional protocol overhead informing the destination about
the relay decision time is not necessary and the loss in DMT due to
decoding error at the relay can be overcome. The difficulty of
analysis comes from that we need to properly manipulate the events
related to not only channel realization(i.e outage) but also
codewords  when applying random coding schemes with finite block
length. We also propose HDAF protocol whcich shares both the
advantages of DDF at low-medium multiplexing gain and MAF at high
multiplexing gain region. It achieves higher(optimal) diversity gain
than DDF protocol, particularly  at high multiplexing gain and at
low multiplexing gain when the number of time slots $M$ is small.
HDAF also outperforms MAF at low-medium multiplexing gain region
when $M$ is reasonable large enough. Finally, we have investigated
the variant of HDAF with reduced complexity by allowing the relay to
switch to transmission mode only after half of codeword and prove
there is no entailing loss in terms of DMT compared to the original
HDAF.

\appendices
 \section{Proof of Theorem \ref{theorem upperbound}}\label{appendixupper}
 Since $d_{m,R}(r)$ is solely a function of the source-relay links
 and $d_{m,D}(r)$ is a function of relay-destination and
 source-destination links. These two terms can be analyzed
 separately.
 \subsection{Analysis of ${d}_{m,R}(r)$} \label{anlysisdmr}
 We partition the event $\{\mathbf{m}=m\}$ into the set of events
 $A_{I}^{m}$, i.e.$\{\mathbf{m}=m\}=\bigcup\limits_{I}A_{I}^{m}$, where $I$ denotes any nonempty subset of
 $\{1,2\}$.  For notation convenience, we make two extra definitions that are only used in  Appendix \ref{appendixupper} and Appendix  \ref{appenddixlower},  $\overline{\mathcal{O}_{I,R}^{0}}=\phi$, ${\mathcal{O}_{R}^{M}}=\phi$ where $\phi$ denotes the empty set.
Then for $m=1,...,M$, $A_{I}^{m}$ are defined as follows
\begin{equation}
A^{m}_{1}=\left\{ (h_{1}, h_{2}): \mathcal{O}_{1,R}^{m-1}\bigcap \overline{\mathcal{O}_{R}^{m}}
 \right\} \label{eq.15}
\end{equation}

%
 Event $A^{m}_{2}$ is similarly defined with $1$ replaced with $2$.
\begin{equation}
A^{m}_{(1,2)}=\left\{ (h_{1}, h_{2}):
\mathcal{O}_{(1,2),R}^{m-1}\bigcap \overline{\mathcal{O}_{R}^{m}}
 \right\} \label{eq.18}
\end{equation}
$A_{I}^{m}$ denotes the event that source message can be decoded at
slot $m$ and the outage event $\mathcal{O}_{I,R}^{m-1}$ occurs.
%
Note
\begin{equation}
P(A^{m}_{I^{*}})\leq
P(\textbf{m}=m)=P(\bigcup\limits_{I}{A^{m}_I})\leq\sum_{I}P(A^{m}_{I})\doteq
P(A^{m}_{I^{*}})  \label{eq.inqdmtmac}
 \end{equation}
where
$I^{*}=\arg\min\limits_{I}\lim\limits_{\rho'\rightarrow\infty}\frac{-\log{P(A^{m}_{I})}}{\log
\rho'}$. Thus $P(\textbf{m}=m)\doteq P(A^{m}_{I^{*}})$. Since
$|h_{i}|^2$ are exponentially distributed and $\rho\doteq \rho'$, we
compute $P(A_{I}^{m})\doteq\rho^{d_{A_{I}^{m}}(r)}$ as follows:

\subsubsection{Computation of ${d}_{A^{m}_{1}}(r)$, for $M-1 \geq m \geq 2$ }

For $r \geq 0$, $M-1 \geq m \geq 2$
\begin{align}
P({A^{m}_{1}})
=&
\begin{cases}
r>\frac{m}{M}, \;\; d_{A_{1}^{m}}(r)=\infty \\
r\leq\frac{m}{M}, \;\; d_{A_{1}^{m}}(r)=1-\frac{Mr}{2(m-1)}
\end{cases} \label{eq.22}
\end{align}
Due to the symmetry,  ${d}_{A^{m}_{2}}(r)={d}_{A^{m}_{1}}(r)$  is
clear.

\subsubsection{Computation of ${d}_{A^{m}_{1,2}}(r)$, for $M-1 \geq m \geq 2$}
%

\begin{align}
P({A^{m}_{1,2}})
=&
\begin{cases}
r>\frac{m}{M}, \;\; d_{A_{1,2}^{m}}(r)=\infty \\
\frac{m}{M} \geq r >\frac{m-1}{M},
\;\;d_{A_{1,2}^{m}}(r)=0\\
 \frac{m-1}{M} \geq r, \;\; d_{A_{1,2}^{m}}(r)=2\left(1-{\frac{Mr}{m-1}}\right) \\
\end{cases} \label{eq.23}
\end{align}
\subsubsection{Computation of  $d_{m=1,R}(r)$} 
%
%
\begin{align}
P({\textbf{m}=1})
=&
\begin{cases}
r>\frac{m}{M}, \;\; {d}_{m=1,R}(r)=\infty \\
\frac{m}{M}\geq r, \;\;d_{m=1,R}(r)=0
\end{cases} \label{eq.24}
\end{align}

\subsubsection{Computation of  ${d}_{m=M,R}(r)$} 
%
%
%

\begin{align}
P (A_{1}^{M})
=&
\begin{cases}
r>\frac{2(M-1)}{M}, \;\; d_{A_{1}^{M}}(r)=0 \\
\frac{2(M-1)}{M}\geq r,
\;\;d_{A_{1}^{M}}(r)={1-\frac{M\frac{r}{2}}{M-1}}
\end{cases} \label{eq.28}
\end{align}

$d_{A_{2}^{M}}(r)=d_{A_{1}^{M}}(r)$ is clear.

\begin{align}
P(A_{1,2}^{M}) 
=&
\begin{cases}
r>\frac{(M-1)}{M}, \;\; d_{A_{1,2}^{M}}(r)=0 \\
\frac{(M-1)}{M}\geq r,
\;\;d_{A_{1,2}^{M}}(r)=2\left({1-\frac{Mr}{M-1}}\right)
\end{cases} \label{eq.29}
\end{align}

Collecting the results (\ref{eq.22})-(\ref{eq.29}), overall, we
obtain

\begin{align}
{d}_{m,R}(r) =
\begin{cases}
0 \leq r< \frac{2(m-1)}{3M}, \;\; d_{m,R}(r)=1-\frac{Mr}{2(m-1)} \\
\frac{2(m-1)}{3M} \leq r< \frac{(m-1)}{M} , \;\; {d}_{m,R}(r)=2\left(1-\frac{Mr}{(m-1)}\right) \\
 \frac{(m-1)}{M} \leq r \leq \frac{m}{M},
\;\;{d}_{m,R}(r)=0 \\
 \frac{m}{M} < r \leq 1 ,
\;\;{d}_{m,R}(r)=\infty
\end{cases} \label{eq.30}
\end{align}

\subsection{Analysis of ${d}_{m,D}(r)$} \label{anlysisdmd}
Similar to (\ref{eq.inqdmtmac}), $P_{out}^{m}(r)\doteq
P_{\mathcal{O}_{I^{*},D}^{m}}(r)$
%
%
%
the DMT $d_{\mathcal{O}_{I,D}^{m}}(r)$ has been derived in
\cite{optimality_ARQ},  we provide  results here,  let
$f=\frac{m}{M}$,
\begin{equation}
d_{\mathcal{O}_{1,D}^{m}}(r)=
\begin{cases}
\frac{1}{2}>f \geq 0, \;\;\;\;\;\;\;\;\;\;  2-r \\
 1-\frac{r}{2}> f \geq \frac{1}{2} ,\;\;\; 2-\frac{r}{2(1-f)} \\
1\geq f \geq 1-\frac{r}{2},   \;\;\;\; \frac{(2-r)}{2f}
\end{cases}\label{eq.36}
\end{equation}

and

for $\frac{1}{3}>r\geq 0$,
\begin{equation}
d_{\mathcal{O}_{(1,2),D}^{m}}(r)=
\begin{cases}
\frac{2}{3}>f \geq 0 , \;\;\;\;\;\;\;\;\; 3(1-r)\\
1-r>f \geq\frac{2}{3} , \;\;\;   3-\frac{r}{1-f}\\
1\geq f \geq 1-r , \;\;\;\;\;  2\frac{1-r}{f}
\end{cases}\label{eq.37}
\end{equation}

for $1 \geq r \geq \frac{1}{3}$,

\begin{equation}
d_{\mathcal{O}_{(1,2),D}^{m}}(r)=
\begin{cases}
\frac{2}{3}>f\geq0  , \;\;\; 3(1-r) \\
1\geq f \geq \frac{2}{3}, \;\;\;\;  2\frac{1-r}{f}
\end{cases}\label{eq.38}
\end{equation}
Note
$d_{m,D}(r)=\min\left\{d_{\mathcal{O}_{(1,2),D}^{m}}(r),d_{\mathcal{O}_{1,D}^{m}}(r),d_{\mathcal{O}_{2,D}^{m}}(r)\right\}$.
By some manipulation, we obtain,
\\

 for $\frac{1}{2}>r \geq 0$,
\begin{equation}
d_{m,D}(r)=
\begin{cases}
\frac{1}{2}>f\geq0 , \;\;\;\;\;\;\;\;\;\;\;  2-r\\
1-\frac{r}{2}>f\geq \frac{1}{2} , \;\;\;\;  2-\frac{r}{2(1-f)}\\
1 \geq f\geq 1-\frac{r}{2}, \;\;\;\;   \frac{2-r}{2f}\\
\end{cases}\label{eq.40}
\end{equation}
 for $\frac{16}{25}>r \geq \frac{1}{2}$,
\begin{equation}
d_{m,D}(r)=
\begin{cases}
 \frac{5r-2}{2(3r-1)}>f\geq 0, \;\;\;\;\;\;\;\;\;  3(1-r)\\
1-\frac{r}{2}>f\geq \frac{5r-2}{2(3r-1)}, \;\;  2-\frac{r}{2(1-f)} \\
1 \geq f\geq 1-\frac{r}{2}  , \;\;\;\;\;\;\;\;\;\;\;\;  \frac{2-r}{2f}\\
\end{cases}\label{eq.41}
\end{equation}
 for $\frac{2}{3}>r \geq \frac{16}{25}$,
\begin{equation}
d_{m,D}(r)=
\begin{cases}
\frac{5r-2}{2(3r-1)}>f\geq 0, \;\;\;  3(1-r)\\
\frac{2-\frac{5r}{4}-\sqrt{r\left(\frac{25r}{16}-1\right)}}{2}>f\geq \frac{5r-2}{2(3r-1)}  , \;\;\;\;2-\frac{r}{2(1-f)}\\
\frac{2-\frac{5r}{4}+\sqrt{r\left(\frac{25r}{16}-1\right)}}{2}>f\geq \frac{2-\frac{5r}{4}-\sqrt{r\left(\frac{25r}{16}-1\right)}}{2} , \;\;\;\; \frac{2(1-r)}{f} \\
1-\frac{r}{2}>f\geq \frac{2-\frac{5r}{4}+\sqrt{r\left(\frac{25r}{16}-1\right)}}{2} , \;\;\;\;  2-\frac{r}{2(1-f)}\\
1 \geq f\geq 1-\frac{r}{2}  , \;\;\;\; \frac{2-r}{2f}\\
\end{cases}\label{eq.42}
\end{equation}
 for $1 \geq r \geq \frac{2}{3}$,
\begin{equation}
d_{m,D}(r)=
\begin{cases}
 \frac{2}{3}>f\geq 0, \;\;\; 3(1-r) \\
1\geq f\geq \frac{2}{3} , \;\;\; \frac{2(1-r)}{f} \\
\end{cases}\label{eq.43}
\end{equation}

 \section{PROOF OF THEOREM \ref{theroem achivable}}\label{appenddixlower}
We now turn to the analysis of the  term
$P(\mathbf{m}=m)P(E_{r}|\mathbf{m}=m)$ in (\ref{eq.48}). First, we
upper bound these two terms by extending the techniques in
\cite{coding_and_decoding_DDF} to our MARC model. Then , we average
the product of these two terms  over the Gaussian random ensemble

\subsection{Analysis of $P(\bold{m}=m)$ }
Let $\mathcal{U}_{m}=\bigcup_{w=1}^{\rho^{rMT}}\mathcal{S}_{m}(w)$,
where $w=(w_{1},w_{2})$. Averaged with respect to the random coding
ensemble, we may choose without loss of generality $w_{1}=1,w_{2}=1$
as the reference transmitted message. Let
$\Delta\bold{x}_{i,0}^{m}(w_{i})=\bold{x}_{i,0}^{m}(w_{i})-\bold{x}_{i,0}^{m}(1)$,
as \cite{coding_and_decoding_DDF}, with $h$ replaced by
$(h_{1},h_{2})$, for $1 \leq m \leq M$, we have the following
results 
\begin{equation}
\begin{split}
 P(\bold{m}=m)\leq & P\left(\{(h_{1},h_{2})\notin
\mathcal{O}_{R}^{m-1}\},\{\bold{y}_{r,0}^{m-1}\notin
\mathcal{U}_{m-1}\}\right) \\
+ & P\left(\{(h_{1},h_{2})\in
\mathcal{O}_{R}^{m-1}\},\{(h_{1},h_{2})\notin
\mathcal{O}_{R}^{m}\}\right)
\end{split}\label{eq.51}
\end{equation}
where the definition $\overline{\mathcal{O}_{R}^{0}}=\phi$,
${\mathcal{O}_{R}^{M}}=\phi$  are again used.

For $1 < m < M$,  given a  $(h_{1},h_{2})$ and a codebook,
\begin{equation}
\begin{split}
& P\left(\{\bold{y}_{r,0}^{m}\notin \mathcal{U}_{m}\}\right) \leq
(1+\delta)^{mT} e^{-mT\delta}
\\
&
+\sum_{w\neq(1,1)}I_{d}\left\{|h_{1}\Delta\bold{x}_{1,0}^{m}(w_{1})+h_{2}\Delta\bold{x}_{2,0}^{m}(w_{2})|^{2}
\leq 4mT(1+\delta)\sigma_{n}^{2}\right\}
\end{split}\label{eq.53}
\end{equation}
where $I_{d}(\;)$ is the indicator function. Different from
\cite{coding_and_decoding_DDF}, we do \textsl{ not} average over the
random coding ensemble here from (\ref{eq.53}). Instead,  we  use
the following inequality
\begin{equation}
\begin{split}
I_{d}\left(|h_{1}\Delta\bold{x}_{1,0}^{m}(w_{1})+h_{2}\Delta\bold{x}_{2,0}^{m}(w_{2})|^{2}\leq
4mT(1+\delta)\sigma_{n}^{2}\right) \\
\leq
e^{1}exp\left\{-\frac{|h_{1}\Delta\bold{x}_{1,0}^{m}(w_{1})+h_{2}\Delta\bold{x}_{2,0}^{m}(w_{2})|^{2}}{4mT(1+\delta)\sigma_{n}^{2}}\right\}
\end{split}\label{eq.54}
\end{equation}
Then we turn to upper bound the term
$P(E|\overline{E}_{r},\bold{m}=m)$
\subsection{ Analysis of $P(E|\overline{E}_{r},\bold{m}=m)$}
We consider the GLRT decoder at the destination. This decoder has no
knowledge of decision time $\bold{m}$. Thus the receiver has to
decode both the decision time and source information message.  The
$P(E|\overline{E}_{r},\bold{m}=m)$ is upper bounded by
\begin{equation}
\begin{split}
P(E|\overline{E}_{r},\bold{m}=m)& \leq P(\{(g_{1}, g_{2}, g_{r}) \in
O_{D}^{m}\})\\
& + P(E, \{(g_{1}, g_{2}, g_{r}) \notin
O_{D}^{m}\}|\overline{E}_{r},\bold{m}=m)
\end{split}
 \label{eq.extrag}
\end{equation}
 and similar to \cite{coding_and_decoding_DDF}, we have
\begin{equation}
\begin{split}
&P(\{(1,1)\rightarrow \tilde{w}\}|\overline{E}_{r}, \bold{m}=m)\\
&\leq \sum_{m'=1}^{M} P\left(p(\mathbf{y}^{M}_{d,0}|1,m)\leq
p(\mathbf{y}^{M}_{d,0}|\tilde{w},m')|\overline{E}_{r},
\bold{m}=m\right)
\end{split}\label{eq.55}
\end{equation}
and for $m'\geq m$, (the case for $m'<m$ can be derived by
interchanging the role of  $m'$ and $m$, thus omitted here ), the
term inside the sum (\ref{eq.55}) can be upper bounded by
\begin{equation}
P\left(p(\mathbf{y}^{M}_{d,0}|1,m)\leq
p(\mathbf{y}^{M}_{d,0}|\tilde{w},m')|\overline{E}_{r},
\bold{m}=m\right) \leq  e^{-|\bold{z}_{m'}|^{2}/(4\sigma_{v}^{2})}
\label{eq.55_1}
\end{equation}
\[
\bold{z}_{m'}\triangleq
\begin{bmatrix}                
  g_{1}\Delta\bold{x}_{1,0}^{m}(\tilde{w}_{1})  + g_{2}\Delta\bold{x}_{2,0}^{m}(\tilde{w}_{2})\\
  g_{1}\Delta\bold{x}_{1,m}^{m'}(\tilde{w}_{1})  + g_{2}\Delta\bold{x}_{2,m}^{m'}(\tilde{w}_{2})+ g_{r}\bold{x}_{r,m}^{m'}(1,1) \\
 g_{1}\Delta\bold{x}_{1,m'}^{M}(\tilde{w}_{1})  + g_{2}\Delta\bold{x}_{2,m'}^{M}(\tilde{w}_{2})+ g_{r}\Delta\bold{x}_{r,m'}^{M}(\tilde{w})
\end{bmatrix}
\]
\subsection{Averaged over the Gaussian random ensemble} \label{sec.avg}
We are ready to average the term $
P(\bold{m}=m)P(E|\overline{E}_{r},\bold{m}=m)$ over the Gaussian
random ensemble. From (\ref{eq.51}),(\ref{eq.extrag}), for $1\leq m
\leq M$, it can be
upper bounded by sum of the four terms discussed below, we will show that the four terms have exponents
larger than or equal to  $d_{m,D}(r)+d_{m,R}(r)$ : \\
\emph{(1)} {$P(\{(g_{1}, g_{2}, g_{r}) \in O_{D}^{m}\})
P\left(\{(h_{1},h_{2})\in
\mathcal{O}_{R}^{m-1}\},\{(h_{1},h_{2})\notin
\mathcal{O}_{R}^{m}\}\right)$}:

It is clear this term has the exponent of $d_{m,D}(r)+d_{m,R}(r)$.
\\  \emph{(2)} {$P(E, \{(g_{1}, g_{2},
g_{r}) \notin O_{D}^{m}\}|\overline{E}_{r},\bold{m}=m)
P\left(\{(h_{1},h_{2})\notin \mathcal{O}_{R}^{m-1}\},\right.$\\
$ \left.  \;\;\; \{\bold{y}_{r,0}^{m-1}\notin
\mathcal{U}_{m-1}\}\right)$}: This term is more involved.  For the
case $m=1$, the term \emph{(2)} becomes to zero. For $1<m \leq M$,
given $h_{i}$, $g_{i}$, from (\ref{eq.51}), (\ref{eq.54}),
(\ref{eq.55}), (\ref{eq.55_1}),  we have
\begin{align}
&P\left(\{\bold{y}_{r,0}^{m-1}\notin \mathcal{U}_{m-1}\}\right)\cdot
P(\{(1,1)\rightarrow \tilde{w}\}|\overline{E}_{r},
\bold{m}=m)  \label{eq.57_1}\\
& \leq
\left[\sum_{w\neq(1,1)}e^{1}exp\left\{-\frac{|h_{1}\Delta\bold{x}_{1,0}^{m-1}(w_{1})+h_{2}\Delta\bold{x}_{2,0}^{m-1}(w_{2})|^{2}}{\left(4(m-1)T(1+\delta)\sigma_{n}^{2}\right)}\right\}\right. \notag\\
&\left.\;\;\;\;\;+(1+\delta)^{(m-1)T }e^{-(m-1)T\delta}\right]\cdot
\sum_{m'=1}^{M}
e^{-|\bold{z}_{m'}|^{2}/(4\sigma_{v}^{2})} \label{eq.57_2} \\
& \dot{\leq}
\sum_{w\neq(1,1)}e^{1}exp\left\{-\frac{|h_{1}\Delta\bold{x}_{1,0}^{m-1}(w_{1})+h_{2}\Delta\bold{x}_{2,0}^{m-1}(w_{2})|^{2}}{\left(4(m-1)T(1+\delta)\sigma_{n}^{2}\right)}\right\}
\notag \\
&\cdot \sum_{m'=1}^{M} e^{-|\bold{z}_{m'}|^{2}/(4\sigma_{v}^{2})}\label{eq.57_3} \\
& \dot{\leq}
\sum_{w\neq(1,1)}exp\left\{-\frac{|h_{1}\Delta\bold{x}_{1,0}^{m-1}(w_{1})+h_{2}\Delta\bold{x}_{2,0}^{m-1}(w_{2})|^{2}}{\left(4(m-1)T(1+\delta)\sigma_{n}^{2}\right)}\right\}
\notag \\
&\cdot \sum_{m'=1}^{M}
e^{-|\bold{z}_{m'}|^{2}/\left(4(m-1)T(1+\delta)\sigma_{v}^{2}\right)} \label{eq.57_4} \\
&=\sum_{w\neq\{(1,1),\tilde{w}\}}
exp\left\{-\frac{|h_{1}\Delta\bold{x}_{1,0}^{m-1}(w_{1})+h_{2}\Delta\bold{x}_{2,0}^{m-1}(w_{2})|^{2}}{\left(4(m-1)T(1+\delta)\sigma_{n}^{2}\right)}\right\} \notag \\
& \cdot \sum_{m'=1}^{M}
e^{-|\bold{z}_{m'}|^{2}/\left(4(m-1)T(1+\delta)\sigma_{v}^{2}\right)}
+\sum_{m'=1}^{M}
e^{-|\bold{z}'_{m'}|^{2}/\left(4(m-1)T(1+\delta)\sigma_{v}^{2}\right)}
 \label{eq.57_5}
\end{align}
where
\[ \bold{z}'^{T}_{m'}\triangleq
\begin{bmatrix}                
 \frac{\sigma_{v}}{\sigma_{n}} \left(h_{1}\Delta\bold{x}_{1,0}^{m-1}(\tilde{w}_{1})  +
  h_{2}\Delta\bold{x}_{2,0}^{m-1}(\tilde{w}_{2})\right) &
   \bold{z}_{m'}
\end{bmatrix}
\]
 Note by separating the term of $\tilde{w}$ from the first sum in (\ref{eq.57_5}),
 then  the product of the first two sum in (\ref{eq.57_5}) can now be averaged separately.
 We first consider the single term
 \begin{equation}
 e^{-|\bold{z}'_{m'}|^{2}/\left(4(m-1)T(1+\delta)\sigma_{v}^{2}\right)}
 \label{eq.57_6}
 \end{equation}
 Let
 \[
 \bold{X}_{i}=[\Delta x_{1,i}(\tilde{w}_{1}), \Delta x_{2,i}(\tilde{w}_{2}), x_{r,i}(1,1),\Delta
 x_{r,i}(\tilde{w})]^{T}, \;\;\;\; 1\leq i \leq MT
 \]
 and
 \[
  \bold{X}=[\bold{X}_{1}^{T},\bold{X}_{2}^{T},...\bold{X}_{MT}^{T}]^{T}
 \]
It can be verified that
$|\bold{z}'_{m}|^{2}=\bold{X}^{H}\bold{R}\bold{X}$, where $\bold{R}$
is a block diagonal matrix of the form
$\bold{R}=diag\left(\bold{R}_{1},...\bold{R}_{MT}\right)$. \\ For $1
\leq i \leq (m-1)T$, \\
\[
\bold{R}_{i} =\bold{C}^{H}\bold{C}
\]
where $h_{i}'=\frac{\sigma_{v}^{2}}{\sigma_{n}^{2}}h_{i}$ and
\[
\bold{C}=
\begin{bmatrix}
  h_{1}'  & h_{2}' &  0& 0\\
  g_{1}  &  g_{2}  & 0 & 0\\
  0 &0&0&0 \\
  0 &0&0&0
\end{bmatrix}
\]
 \\For
$(m-1)T+1 \leq i \leq mT$
\[
\bold{R}_{i}=\bold{g}_{a}^{H}*\bold{g}_{a}\;,
\;\;\;\;\;\;\bold{g}_{a}=[g_{1},g_{2},0, 0]
\]
For $mT+1 \leq i \leq m'T$
\[
\bold{R}_{i}=\bold{g}_{b}^{H}*\bold{g}_{b}\;, \;\;\;\;\;\;
\bold{g}_{b}=[g_{1},g_{2},g_{r}, 0]
\]
For $m'T+1 \leq i \leq MT$
\[
\bold{R}_{i}=\bold{g}_{c}^{H}*\bold{g}_{c}\;, \;\;\;\;\;\;
\bold{g}_{c}=[g_{1},g_{2}, 0,g_{r}]
\]
Let $\bold{R}^{s}=diag(\bold{R}_{1}^{s},...,\bold{R}_{MT}^{s})$,
$\bold{D}=diag(\bold{D}_{1},...,\bold{D}_{MT})$ where
$\bold{R}_{i}^{s}$ denotes $\bold{R}_{i}$ with the $s$-th row and
the $s$-th column replaced by zero vector.(when $s=0$, nothing is
changed). $\bold{D}_{i}=diag(2P,2P,P,2P)$. $\mathbf{C}^{s}$ is
similarly defined. Averaging (\ref{eq.57_6}) over the Gaussian
random ensemble, we have
\begin{equation}
\frac{1}{\det\left(\bold{I}+\frac{\bold{D}\bold{R}^{s}}{\left(4(m-1)T(1+\delta)\sigma_{v}^{2}\right)}
\right)} \label{eq.58}
\end{equation}
Let $k=\frac{1}{\left((m-1)T(1+\delta)\right)}$, we have\\
\\
\emph{(I )} For $\{\tilde{w}_{1}\neq 1, \tilde{w}_{2}\neq 1\}$,
$s=0$ ,
\begin{align}
\frac{1}{\det\left(\bold{I}+\frac{k\bold{D}\bold{R}^{s}}{4}\right)}
= & \frac{1}{\det\left(\bold{I}+\frac{k\rho\bold{C}^{H}\bold{C}}{2} \right)^{(m-1)T}} \notag\\
 \cdot &
 \frac{1}{\left[1+\frac{k\rho(|g_{1}|^{2}+|g_{2}|^{2})}{2}\right]^{T}} \notag \\
  \cdot &
  \frac{1}{\left[1+\frac{k\rho(2|g_{1}|^{2}+2|g_{2}|^{2}+|g_{r}|^{2})}{4}\right]^{(m'-m)T}} \notag \\
   \cdot &
   \frac{1}{\left[1+\frac{k\rho(|g_{1}|^{2}+|g_{2}|^{2}+|g_{r}|^{2})}{2}\right]^{(M-m')T}}
   \notag \\
  \dot{ \leq} &
    \frac{1}{\left[1+\rho(|h_{1}|^{2}+|h_{2}|^{2}+|g_{1}|^{2}+|g_{2}|^{2})\right]^{(m-1)T}}  \notag \\
 \cdot & \frac{1}{\left[1+\rho(|g_{1}|^{2}+|g_{2}|^{2})\right]^{T}} \notag \\
   \cdot &
   \frac{1}{\left[1+\rho(|g_{1}|^{2}+|g_{2}|^{2}+|g_{r}|^{2})\right]^{(M-m)T}}
\label{eq.59}
\end{align}
Notice that (\ref{eq.59}) does not depend on $m'$ and follows from
the inequality
\[
\det\left(\bold{I}+\frac{k\rho\bold{C}^{H}\bold{C}}{2} \right) \geq
\left[1+\frac{k\rho}{2}(|h_{1}|^{2}+|h_{2}|^{2}+|g_{1}|^{2}+|g_{2}|^{2})\right]
\]
Define $|h_{i}|^{2}=\rho^{-\alpha_{i}}$,
$|g_{i}|^{2}=\rho^{-\beta_{i}}$. Use the union bound for
$\mathcal{\overline{B}}_{m} $, summing over all $m'=1,...,M$ and
over all messages $\{ \tilde{w}_{1}\neq 1, \tilde{w}_{2}\neq 1 \}$,
and  average  over the channel realizations where $\{(h_{1},h_{2})
\notin \mathcal{O}_{R}^{m-1}\}$, $\{(g_{1},g_{2}, g_{r})\notin
\mathcal{O}_{D}^{m}\}$,  use the techniques developed in
\cite{on_the_achievableDMT}, we obtain the exponent of this
correlated term as
\begin{align}
 d_{m,cor}^{s=0}(r)&=\inf_{\substack{{\alpha_{i},\beta_{i}\geq
0},\\{\overline{\mathcal{O}_{R}^{m-1}} \bigcap
\overline{\mathcal{O}_{D}^{m}}}}}\left\{
\alpha_{1}+\alpha_{2}+\beta_{1}+\beta_{2}+\beta_{r}+Tf_{m}^{s=0}(\alpha_{i},
\beta_{i}, r)\right\} \notag\\
& \geq \inf_{\substack{{\alpha_{i},\beta_{i}\geq 0},\\
\overline{\mathcal{O}_{R}^{m-1}}}}\left\{
\alpha_{1}+\alpha_{2}+\beta_{1}+\beta_{2}+\beta_{r}+Tf_{m}^{s=0}(\alpha_{i},
\beta_{i}, r)\right\} \label{eq.63_0}\\
& \geq \inf_{\substack{{\alpha_{i},\beta_{i}\geq 0},\\
\overline{\mathcal{O}_{R}^{m-1}}}}\left\{
\alpha_{1}+\alpha_{2}+Tf_{m}^{s=0}(\alpha_{i},
\beta_{i}, r)\right\} \label{eq.63_1}\\
&+ \inf_{\substack{{\alpha_{i},\beta_{i}\geq 0},\\
\overline{\mathcal{O}_{R}^{m-1}}}}\left\{
\beta_{1}+\beta_{2}+\beta_{r}+Tf_{m}^{s=0}(\alpha_{i},
\beta_{i}, r)\right\} \label{eq.63_2} \\
& \geq d_{m,R}(r)+d_{m,D}(r) \label{eq.63_3}
\end{align}
where \begin{equation}
\begin{split}
&f_{m}^{s=0}(\alpha_{i},\beta_{i},r)\\
& = (m-1) \max {[(1-\alpha_{1})^{+}
,(1-\alpha_{2})^{+},(1-\beta_{1})^{+},
(1-\beta_{2})^{+}]} \\
& +\max[(1-\beta_{1})^{+}, (1-\beta_{2})^{+}] \\
& +(M-m) \max {[(1-\beta_{1})^{+}, (1-\beta_{2})^{+},
(1-\beta{r})^{+}]-rM } \label{eq.62}
\end{split}
\end{equation}
(\ref{eq.63_1}), (\ref{eq.63_2}) follows form the fact that the
infimum of (\ref{eq.63_0}) occurs where $f_{m}^{s=0}(\alpha_{i},
\beta_{i}, r)\rightarrow 0$ if $T$ is large enough(i.e. $T\geq4$).
Notice the event
\begin{equation}
\left\{ (\alpha_{i}): (m-1) \max {[(1-\alpha_{1})^{+}
,(1-\alpha_{2})^{+}]}-rM>0 \right\} \in
\overline{\mathcal{O}_{R}^{m-1}}
\end{equation}
It can be checked that  $(\ref{eq.63_1})\geq d_{m,R}(r)$, and
$(\ref{eq.63_2}) \geq d_{m,D}(r)$. \\
\emph{(II)} For $\{\tilde{w}_{1}\neq 1, \tilde{w}_{2}= 1\}$, $s=2$ ,

\begin{align}
\frac{1}{\det\left(\bold{I}+\frac{k\bold{D}\bold{R}^{s}}{4}\right)}
 \doteq &
    \frac{1}{\left[1+\rho(|h_{1}|^{2}+|g_{1}|^{2})\right]^{(m-1)T}}  \notag \\
 \cdot & \frac{1}{\left[1+\rho|g_{1}|^{2}\right]^{T}} \notag \\
   \cdot &
   \frac{1}{\left[1+\rho(|g_{1}|^{2}+|g_{r}|^{2})\right]^{(M-m)T}}
\label{eq.2_59}
\end{align}
Apply  similar arguments, we obtain the exponent of this correlated
term as
\begin{align}
 d_{m,cor}^{s=2}(r)&=\inf_{\substack{{\alpha_{i},\beta_{i}\geq
0},\\{\overline{\mathcal{O}_{R}^{m-1}} \bigcap
\overline{\mathcal{O}_{D}^{m}}}}}\left\{
\alpha_{1}+\beta_{1}+\beta_{r}+Tf_{m}^{s=2}(\alpha_{i},
\beta_{i}, r)\right\} \notag\\
& \geq \inf_{\substack{{\alpha_{i},\beta_{i}\geq 0},\\
\overline{\mathcal{O}_{R}^{m-1}}}}\left\{
\alpha_{1}+\beta_{1}+\beta_{r}+Tf_{m}^{s=2}(\alpha_{i},
\beta_{i}, r)\right\} \label{eq.2_63_0}\\
& \geq \inf_{\substack{{\alpha_{i},\beta_{i}\geq 0},\\
\overline{\mathcal{O}_{R}^{m-1}}}}\left\{
\alpha_{1}+Tf_{m}^{s=2}(\alpha_{i},
\beta_{i}, r)\right\} \label{eq.2_63_1}\\
&+ \inf_{\substack{{\alpha_{i},\beta_{i}\geq 0},\\
\overline{\mathcal{O}_{R}^{m-1}}}}\left\{
\beta_{1}+\beta_{r}+Tf_{m}^{s=2}(\alpha_{i},
\beta_{i}, r)\right\} \label{eq.2_63_2} \\
& \geq d_{m,R}(r)+d_{m,D}(r) \label{eq.2_63_3}
\end{align}
where \begin{equation}
\begin{split}
f_{m}^{s=2}(\alpha_{i},\beta_{i},r)& = (m-1) \max {[(1-\alpha_{1})^{+} ,(1-\beta_{1})^{+}]} \\
& +\max(1-\beta_{1})^{+} \\
& +(M-m) \max {[(1-\beta_{1})^{+}, (1-\beta{r})^{+}]-\frac{r}{2}M }
\label{eq.2_62}
\end{split}
\end{equation}
Note
\begin{equation}
\left\{ \alpha_{1}: (m-1) \max {(1-\alpha_{1})^{+}}-\frac{r}{2}M>0
\right\} \in \overline{\mathcal{O}_{R}^{m-1}}
\end{equation}
The case for  $\{\tilde{w}_{1}= 1, \tilde{w}_{2}\neq 1\}$, $s=1$ is
similar and omitted here. Thus this correlated term can be ignored
with respect  to the term \emph{(1)} in DMT analysis.

For averaging the first sum in (\ref{eq.57_5}),  define
\[
\bold{H}'^{s}_{i}=\begin{cases} \bold{h}'^{H}\bold{h}'  \;\;\; 1 \leq i \leq (m-1)T \\
\textbf{0} \;\;\;\;\;\;\;\;\;\; otherwise
\end{cases}
\]
where
\[
\bold{h}'=\begin{bmatrix} h'_{1} & h'_{2} &0 & 0
\end{bmatrix}
\]
and $s$ is chosen according to $w$ as before, we have
\begin{equation}
\left(\sum_{w\neq\{(1,1),\tilde{w}\}}\frac{1}{\det\left(\bold{I}+\frac{k\bold{D}\bold{H}^{s}}{4}\right)}\right)
\label{eq.64}
\end{equation}
Averaged by $\{(h_{1},h_{2}) \notin \mathcal{O}_{R}^{m-1}\}$,
$\{(g_{1},g_{2}, g_{r})\notin \mathcal{O}_{D}^{m}\}$, the exponent
$d_{m,(h_{i})}(r)$ can be derived.

To average the second sum in (\ref{eq.57_5}), let
$\bold{Q}^{s}=\bold{R}^{s}$, except for $1 \leq i \leq (m-1)T$,
$\bold{Q}^{s}_{i}= \bold{R}^{s}_{mT}$, use the union bound over
$\tilde{w}$, by averaging the term
\begin{equation}
\left(\sum_{\tilde{w}\neq(1,1)}\frac{1}{\det\left(\bold{I}+\frac{k\bold{D}\bold{Q}^{s}}{4}\right)}\right)
\label{eq.65}
\end{equation}
over  $\{(h_{1},h_{2}) \notin \mathcal{O}_{R}^{m-1}\}$,
$\{(g_{1},g_{2}, g_{r})\notin \mathcal{O}_{D}^{m}\}$, we obtain the
exponent $d_{m,(g_{i})}(r)$.  Analogous to Appendix
\ref{anlysisdmr}, Appendix \ref{anlysisdmd}, it can be shown that
$d_{m,(h_{i})}(r)\dot{\geq} d_{m,R}(r)$, $d_{m,(g_{i})}(r)\doteq
d_{m,D}(r)$. Collecting all the results above, the term \emph{(2)}
can be
ignored compared to the term \emph{(1)} in DMT analysis. \\
\emph{(3)} {$P(\{(g_{1}, g_{2}, g_{r}) \in O_{D}^{m}\})
P\left(\{(h_{1},h_{2})\notin
\mathcal{O}_{R}^{m-1}\},\{\bold{y}_{r,0}^{m-1}\notin
\mathcal{U}_{m-1}\}\right)$}: These two terms can be averaged over
the Gaussian random ensemble separately. It is clear the first term
has the exponent of $d_{m,D}(r)$, and the second has
$d_{m,(h_{i})}(r)$.
\\
\emph{(4)} {$P(E, \{(g_{1}, g_{2}, g_{r}) \notin
O_{D}^{m}\}|\overline{E}_{r},\bold{m}=m) P\left(\{(h_{1},h_{2})\in
\mathcal{O}_{R}^{m-1}\} \right. ,\\ \left.
\;\;\;\{(h_{1},h_{2})\notin \mathcal{O}_{R}^{m}\}\right) $}: These
two terms can be averaged over the Gaussian random ensemble
separately. It is clear the first term has the exponent of
$d_{m,(g_{i})}(r)$, and the second has $d_{m,R}(r)$.
\section{Proof of Theorem \ref{theoremhybrid} }\label{appendixhybrib}
To upper bound the second term on the right hand side of
(\ref{eq.hyb3}), we further partition the error event into two error
events depending on the MAF's outage event $O_{MAF}$ of
$\alpha_{i},\beta_{i}$, where the MAF is being used,
\begin{equation}
\begin{split}
&P\left(E,\left\{\{(h_{1},h_{2})\notin
\mathcal{O}_{R}^{\frac{M}{2}}\},\{\bold{y}_{r,0}^{\frac{M}{2}}\notin
\mathcal{U}_{\frac{M}{2}}\}\right\}\right)= \\
&P\left(E,O_{MAF},\left\{\{(h_{1},h_{2})\notin
\mathcal{O}_{R}^{\frac{M}{2}}\},\{\bold{y}_{r,0}^{\frac{M}{2}}\notin
\mathcal{U}_{\frac{M}{2}}\}\right\}\right)+\\
&P\left(E,\overline{O}_{MAF},\left\{\{(h_{1},h_{2})\notin
\mathcal{O}_{R}^{\frac{M}{2}}\},\{\bold{y}_{r,0}^{\frac{M}{2}}\notin
\mathcal{U}_{\frac{M}{2}}\}\right\}\right)
\end{split}
\label{eq.hyb6}
\end{equation}
Notice we do not apply the standard DMT analysis \cite{diva} here by
simply using $P\left(O_{MAF}\right)$ to upper bound the first term
on the right hand side of (\ref{eq.hyb6}) (denoted as $P^{1st}(\xi)$
for notation ease), otherwise the lower bound of DMT  we obtain for
(\ref{eq.hyb6}) will be the same as  MAF protocol, only
$2-\frac{3r}{2}$, not $2-r$.

For the the analysis of  $P^{1st}_{\xi}$,  note the fact \[ P\left(E
\left |(\alpha_{i},\beta_{i}) \in
O_{MAF},\left\{\{(h_{1},h_{2})\notin
\mathcal{O}_{R}^{\frac{M}{2}}\},\{\bold{y}_{r,0}^{\frac{M}{2}}\notin
\mathcal{U}_{\frac{M}{2}}\}\right\}\right.\right)\leq 1
\]
we have
\[
\begin{split}
&P^{1st}\left({\xi|(\alpha_{i},\beta_{i}) \in O_{MAF}}\right)\\
&\leq P\left( \left.\left\{\{(h_{1},h_{2})\notin
\mathcal{O}_{R}^{\frac{M}{2}}\},\{\bold{y}_{r,0}^{\frac{M}{2}}\notin
\mathcal{U}_{\frac{M}{2}}\}\right\}\right|(\alpha_{i},\beta_{i}) \in
O_{MAF}\right)
\end{split}
\]
Recall (\ref{eq.53}), (\ref{eq.54}),  by averaging over the Gaussian
random codebook, we have
\begin{equation}
\begin{split}
&P^{1st}\left({\xi|(\alpha_{i},\beta_{i}) \in
O_{MAF}}\right)\\
&\dot{\leq}
\sum_{i=1}^{2}\rho^{-\frac{MT}{2}[(1-\alpha_{i})-\frac{r}{2}]}+\rho^{-\frac{MT}{2}[(1-\min(\alpha_{1},\alpha_{2})-r]}
\label{eq.hyb7}
\end{split}
\end{equation}
Finally, averaging the right hand side of (\ref{eq.hyb7}) over
$(\alpha_{i},\beta_{i}) \in
\left\{O_{MAF}\bigcap\left((\alpha_{1},\alpha_{2})\notin
\mathcal{O}_{R}^{\frac{M}{2}}\right)\right\}$, for $T$ is large
enough( but finite), the derivation of $P^{1st}(\xi)\doteq
\rho^{d_{P^{1st}_{\xi}}(r)}$ becomes identical to (\ref{eq.hyb4}),
(\ref{eq.hyb8}) and $d_{P^{1st}_{\xi}}(r)\geq 2-r$. Then we turn to
upper bound the second term on the right hand side of
(\ref{eq.hyb6}), denoted as $P^{2nd}({\xi})$. As (\ref{eq.hyb9}),
$P^{2nd}({\xi})\leq \sum_{I}P^{2nd}({\xi_{I}})$.
Use $P_{MAF}\left(E,\overline{O}_{MAF}\right)$, where MAF is always
used, to upper bound $P^{2nd}(\xi)$, we have
$P^{2nd}(\xi_{(1,2)})\geq 3-r$, however, only $P^{2nd}(\xi_{i})\geq
2-\frac{3r}{2}$ is obtained, $i=1,2$. To get a tighter upper bound
for $P^{2nd}(\xi_{1})$, (\ref{eq.53}), (\ref{eq.54}) are again used
and apply the techniques developed in Appendix \ref{sec.avg} to MAF
protocol, where the equivalent MAF model can be expressed as
\begin{equation}
\mathbf{y}=\begin{bmatrix}
  g_{1}\mathbf{I}_{MT/2}  & \mathbf{0} \\
  g_{2}h_{1}b\mathbf{I}_{MT/2}  & g_{1}\mathbf{I}_{MT/2}\\
\end{bmatrix}\mathbf{x}_{1}+
\begin{bmatrix}
   \mathbf{0} \\
  g_{2}b\mathbf{I}_{MT/2}\\
\end{bmatrix}\mathbf{n}+\mathbf{v} \label{eq.hybMAFmodel}
\end{equation}
and $b$ is chosen to be of exponential order zero \cite{case_MAC},
\cite{on_the_achievableDMT}, which satisfies the relay's
transmission power constraint.
 We have
\begin{equation}
\begin{split}
&P^{2nd}_{1}\left(\{(1,1)\rightarrow \tilde{w}\}\right)\dot{\leq}
e^{-|\mathbf{z}'|^{2}} + \\
& \sum_{w\neq\{(1,1),\tilde{w}\}}
exp\left\{-\frac{|h_{1}\Delta\bold{x}_{1,0}^{\frac{M}{2}}(w_{1})+h_{2}\Delta\bold{x}_{2,0}^{\frac{M}{2}}(w_{2})|^{2}}{\left(2MT(1+\delta)\sigma_{n}^{2}\right)}\right\}e^{-|\mathbf{z}|^{2}}
\end{split}
\label{eq.hyb12}
\end{equation}
where $P^{2nd}_{1}\left(\{(1,1)\rightarrow \tilde{w}\}\right)$
denotes the probability that $(1,1)$ is decoded as $\tilde{w}$ for
type-1 error in the event of the second term in (\ref{eq.hyb6})
given $h_{i},g_{i}$ and $\mathbf{z}$ and $\mathbf{z}'$ are redefined
in this section as
\begin{equation}
\mathbf{z}=\begin{bmatrix}
  \frac{g_{1}}{\sigma_{n}}\mathbf{I}_{MT/2}  & \mathbf{0} \\
  \frac{g_{2}h_{1}b}{\sqrt{\sigma_{n}^{2}+|g_{2}b|^{2}\sigma_{v}^{2}}}\mathbf{I}_{MT/2}  & \frac{g_{1}}{ \sqrt{\sigma_{n}^{2}+|g_{2}b|^{2}\sigma_{v}^{2}}}\mathbf{I}_{MT/2}\\
\end{bmatrix}
\Delta\mathbf{x}_{1}(\tilde{w}) \label{eq.hyb10}
\end{equation}
\begin{equation}
\mathbf{z}'=\begin{bmatrix}
\frac{h_{1}}{\sigma_{n}^{2}}\mathbf{I}_{MT/2}  & \mathbf{0} \\
  \frac{g_{1}}{\sigma_{n}^{2}}\mathbf{I}_{MT/2}  & \mathbf{0} \\
  \frac{g_{2}h_{1}b}{\sqrt{\sigma_{n}^{2}+|g_{2}b|^{2}\sigma_{v}^{2}}}\mathbf{I}_{MT/2}  & \frac{g_{1}}{\sqrt{\sigma_{n}^{2}+|g_{2}b|^{2}\sigma_{v}^{2}}}\mathbf{I}_{MT/2}\\
\end{bmatrix}
\Delta\mathbf{x}_{1}(\tilde{w}) \label{eq.hyb11}
\end{equation}
Apply the union bound and average them over Gaussian random ensemble
and corresponding channel  realizations. For the product term on the
right hand side of (\ref{eq.hyb12}), it can be averaged separately
over the Gaussian random ensemble, resulting in the exponent
$d_{product}(r)$ identical to (\ref{eq.hyb4}). For the first term on
the right hand side of (\ref{eq.hyb12}), we obtain its exponent
equal to
\begin{equation}
\inf_{(\alpha_{i},\beta_{i}) \in
\mathcal{C}}\left\{\alpha_{1}+\alpha_{2}+\beta_{1}+\beta_{2}+\beta_{r}+\frac{MT}{2}k(\alpha_{i},
\beta_{i},r)\right\}\geq2-r
\end{equation}
where
\[
k(\alpha_{i},
\beta_{i},r)=\max\left\{2(1-\beta_{1}),1-(\beta_{r}+\alpha_{1}),2-(\alpha_{1}+\beta_{1})\right\}-r
\]
and
 $\mathcal{C}= \left\{\overline{O}_{MAF}\bigcap\left((\alpha_{1},\alpha_{2})\notin
\mathcal{O}_{R}^{\frac{M}{2}}\right)\right\}$
\bibliographystyle{IEEEtran}
\bibliography{IEEEabrv,reference}

\end{document}